\documentclass[11pt]{article}
\usepackage{graphicx}  
\graphicspath{{./Fig/}{./}}

\usepackage{hyperref} 
\usepackage{amsfonts}
\usepackage{amsmath}
\usepackage{amssymb}
\usepackage{bbm}
\usepackage{amsthm}

\newtheorem{theorem}{Theorem}
\newtheorem{proposition}{Proposition}
\newtheorem{definition}{Definition}
\newtheorem{corollary}{Corollary}

\newtheorem{example}{Example}
\newtheorem{lemma}{Lemma}
\newtheorem{remark}{Remark}

\usepackage{todonotes}
\usepackage{cleveref}
 
\usepackage{subcaption}
\usepackage{quiver}
\usepackage[figuresleft]{rotating}

\usepackage{orcidlink}
\usepackage{fullpage}
\usepackage{hyperref,url}

\def\Mob{\mathrm{Mob}}
\def\barcalL{{\overline{\mathcal{L}}}}
\def\tr{\mathrm{tr}}

\def\eqdef{:=}

\newcommand{\VPM}{\operatorname{VPM}}
\def\barVPM{{\overline{\VPM}}}
\newcommand{\Sym}{\operatorname{Sym}}

\newcommand{\PD}{\operatorname{PD}}
\newcommand{\PSD}{\operatorname{PSD}}
\renewcommand{\dh}{d_H}

\def\eps{\epsilon}
\def\GL{\mathrm{GL}}  
\def\PGL{\mathrm{PGL}} 

\def\diag{\mathrm{diag}}

\def\ds{\mathrm{d}s}
\def\dQ{\mathrm{d}Q}

\def\inner#1#2{{\langle #1,#2\rangle}}
\def\Inner#1#2{{\left\langle #1,#2\right\rangle}}

\def\Coll{\mathrm{Coll}}
\def\Isom{\mathrm{Isom}}
\def\range{\mathrm{range}}
\def\st{{\ :\ }}
\def\bbR{\mathbb{R}}
\def\bbRpos{\mathbb{R}_{>0}}

\def\mattwotwo#1#2#3#4{{\left[\begin{array}{ll}#1 & #2\cr #3 & #4\end{array}\right]}}
\def\cmattwotwo#1#2#3#4{{\left[\begin{array}{cc}#1 & #2\cr #3 & #4\end{array}\right]}}
\def\xmark{{\it $\times$}}

\title{Hilbert geometry of the symmetric positive-definite bicone: 
Application to the geometry of the extended Gaussian family}

\author{Jacek Karwowski\thanks{Equal contribution.}~\orcidlink{0000-0002-8361-2912}\\
Department of Computer Science\\
University of Oxford, UK
\and
Frank Nielsen\footnotemark[1]~\orcidlink{0000-0001-5728-0726}\\
Sony Computer Science Laboratories Inc.\\
Tokyo, Japan
}

\date{} 

\begin{document}

\maketitle

\begin{abstract}
The extended Gaussian family is the closure of the Gaussian family obtained by completing the  Gaussian family with the counterpart elements induced by  degenerate covariance or degenerate precision matrices, or a mix of both degeneracies. 
The parameter space of the extended Gaussian family forms a  symmetric positive semi-definite matrix bicone, i.e. two partial symmetric positive semi-definite matrix cones joined at their bases.
In this paper, we study the Hilbert geometry of such an open  bounded convex symmetric positive-definite  bicone.
We report the closed-form formula for the corresponding Hilbert metric distance and study exhaustively its invariance properties.
We also touch upon potential applications of this geometry for dealing with extended Gaussian distributions. 
\end{abstract}


\section{Introduction}

To first motivate our study, we introduce the bicone closed parameter space of the extended Gaussian family in \S\ref{sec:eGaussian}.
Since we shall consider the Hilbert geometry of the open bicone in this work, we give some background of the Hilbert geometry and its related Birkhoff projective geometry, with applications in \S\ref{sec:HG}.
The organization of the paper with our contributions is summarized in \S\ref{sec:contrib}.

\subsection{The extended Gaussian family}\label{sec:eGaussian}

\begin{figure}%
\includegraphics[width=\columnwidth]{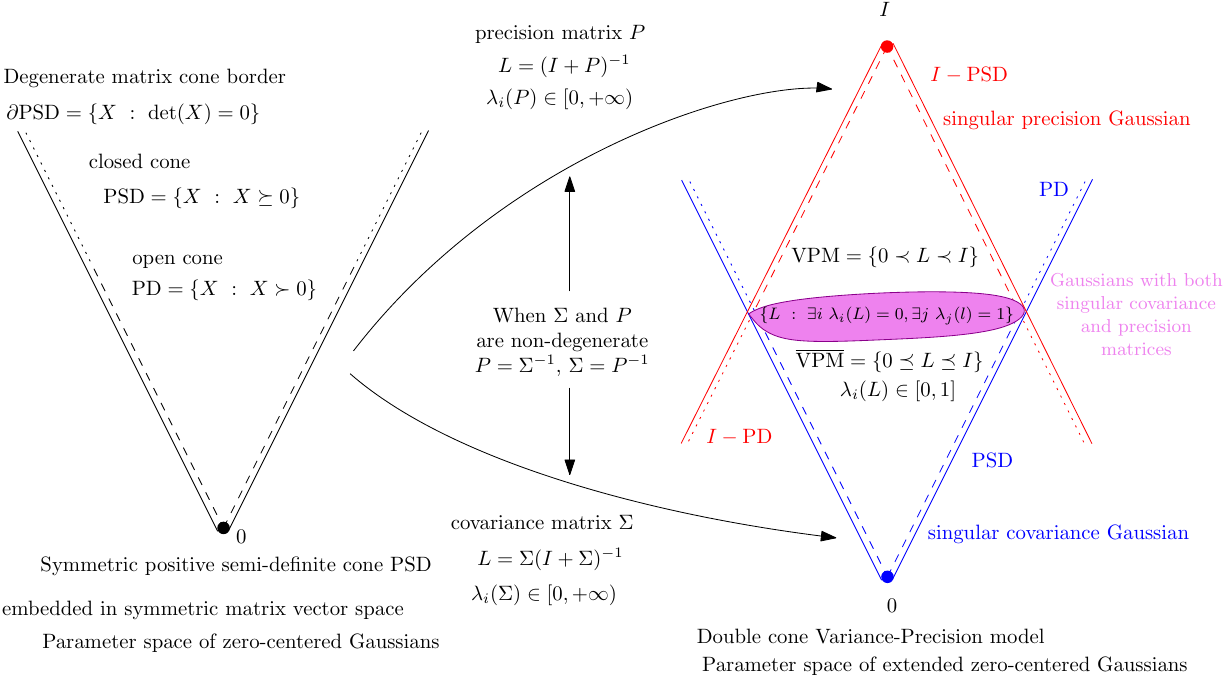}%
\caption{The parameter space of the extended Gaussian family forms a closed double cone.}%
\label{fig:GaussianVPM}%
\end{figure}

Let $N_0(\Sigma)\sim p_{\Sigma}(x)$ denote a $n$-variate Gaussian distribution (normal distribution) centered at the origin with density with respect to the Lebesgue measure given by
$$
p_{\Sigma}(x)= \frac{1}{(2\pi)^{\frac{n}{2}}\, \sqrt{\det(\Sigma)}} \, \exp\left( -\frac{1}{2} \Inner{x}{\Sigma^{-1}x}\right),
$$
where $\inner{x}{x'}=\sum_{i=1}^n x_ix_i'$ denote the Euclidean inner product.
In the following, we assume all Gaussian distributions are centered, and discuss in the perspective section how to consider arbitrary (extended) Gaussian distributions.
 
The parameter space of the Gaussian family $\{N_0(\Sigma)\}$ is the symmetric positive-definite cone (SPD) $\PD=\{X\in\Sym(\bbR,n) \st X\succ 0\}$, where $\succ$ denotes the Loewner partial ordering on the vector space $\Sym(\bbR, n)$ of  $n\times n$ real symmetric matrices (written concisely as $\Sym(n)$ in the remainder): that is, $A\succ B$ if and only if $A-B\succ 0$, i.e., $A-B\in\PD$.
The topologically closed SPD cone is the positive semi-definite (PSD) cone:
$\PSD=\{X\in\Sym(n) \st X\succeq 0\}$ where $A\succeq B$ iff. $A-B$ is PSD, i.e. $A-B\in \PSD$.
A Gaussian $N_0(\Sigma)$ can also be equivalently be parameterized by its precision matrix $P=\Sigma^{-1}$ (also called information matrix~\cite{VarianceInfoMfd-1973}):
 $N_0(\Sigma)=N_0(P)$ with $N_0(P)\sim p_{P}(x)$:
$$
p_{P}(x)= \frac{\sqrt{\det(P)}}{(2\pi)^{\frac{n}{2}}} \, \exp\left( -\frac{1}{2} \Inner{x}{Px}\right).
$$
From the viewpoint of geometry, 
we consider the Gaussian family as a cone manifold equipped with a single global coordinate system (e.g., the covariance or precision coordinate systems which 
form a pair of dual coordinate systems $(\Sigma(\cdot),\PD)$ and  $(P(\cdot),\PD)$   in the setting of information geometry~\cite{GaussianDFS-1996,ohara2019doubly}). We denote by $\PD(n)$ and $\PSD(n)$ the open $n$-dimensional SPD and PSD cones, respectively.

Gaussian distributions with degenerate covariance matrices $\Sigma$ are commonly considered as Gaussian distributions defined on lower-dimensional affine subspaces $\range(\Sigma)=\{\Sigma x \st x\in\bbR^n\}$ of $\bbR^n$. 
James~\cite{VarianceInfoMfd-1973} considered Gaussian distributions with degenerate precision matrices, and more generally proposed to extend the Gaussian family by considering the following reparameterization of Gaussians:
\begin{eqnarray*}
L(\Sigma) &=& \Sigma (I+\Sigma)^{-1},\\
L(P)&=& (I+P)^{-1},
\end{eqnarray*}
and:
\begin{eqnarray*}
\Sigma(L)&=& L(I-L)^{-1},\\
P(L)&=& L^{-1}-I,
\end{eqnarray*}

where $I=I_n$ denotes the $n\times n$ identity matrix.
For full rank covariance/precision matrices, the eigenvalues $\lambda_i(L)$ fall inside $(0,1)$ for $i\in\{1,\ldots,n\}$.
James~\cite{VarianceInfoMfd-1973} first considered the general case of extended Gaussian distributions 
by considering ``Gaussian elements'' (i.e., the counterpart of Gaussians with degenerate covariance and/or precision matrices) lying on the boundary of the bicone\footnote{A word on terminology: A bicone is two cones joined at their bases while a double cone is two cones joined at their apex.}:
that is, extended Gaussians with either degenerate covariance or precision matrices, and Gaussians with both
degenerate covariance and precision matrices.
See Figure~\ref{fig:GaussianVPM} for a schematic illustration.
The variance-precision model $\barVPM$ is the intersection of the two PSD cones and the variance-precision manifold $\VPM$ is the interior of $\barVPM$:
 $\VPM=\barVPM \backslash \partial \barVPM$.
James~\cite{VarianceInfoMfd-1973} called $\VPM$ the variance-information manifold and studied the precision degeneracies of variance-information model $\barVPM$.

The pioneer approach of James was later studied in depth using the framework of category theory in~\cite{Stein-2023}, further highlighting interpretations of the different types of degenerate extended Gaussian distributions and their duality relationships was elucidated.
In a nutshell, Gaussian distributions with degenerate precision matrices encode non-determinism relations in the open stochastic framework of Willems~\cite{willems2012open}, and  play the role of  affine subspace constraints (loosely speaking, those degenerate precision Gaussians can be interpreted as ``uniform distributions'' although the concept of uniform distributions in subspaces with respect to Lebesgue measure does not exist in measure theory).

In 2D, we may visualize the 2D PSD matrix cone $\PSD(2)\in\Sym(2)$ as an equivalent Lorentz cone $\barcalL_3$ (second-order cone having the shape of an ice cream cone) in $\bbR^3$ where
$$
\barcalL_3=\{(t,x,y)\in\bbR^3 \st t\geq \sqrt{x^2+y^2}\}.
$$

The invertible diffeomorphisms $v(Q)\Leftrightarrow Q(v)$ between $\barcalL_3$ and $\VPM$ are given by
$$
v(Q)=\left(t=\frac{a+b}{2},x=\frac{a-b}{2},y=c\right) \in\barcalL_3 
\Leftrightarrow
Q(t,x,y)=\cmattwotwo{t+x}{y}{y}{t-x}=\cmattwotwo{a}{c}{c}{b}
\succeq 0.
$$

Thus $\barVPM(2)$ can be interpreted as the Lorentz bicone obtained by the intersection of $\barcalL_3$ with the reverse translated Lorentz cone $(0,0,1)-\barcalL_3$:
$\barVPM(2) \Leftrightarrow \barcalL_3\cap ((0,0,1)-\barcalL_3)$.
See Figure~\ref{fig:GaussianVPM2} for some renderings of the border of $\barVPM(2)$.\footnote{A video of an animation is available online at \url{https://www.youtube.com/shorts/PEJG8Xtz46c}}

\begin{figure}%
\centering
\begin{tabular}{ccc}
\includegraphics[width=0.28\columnwidth]{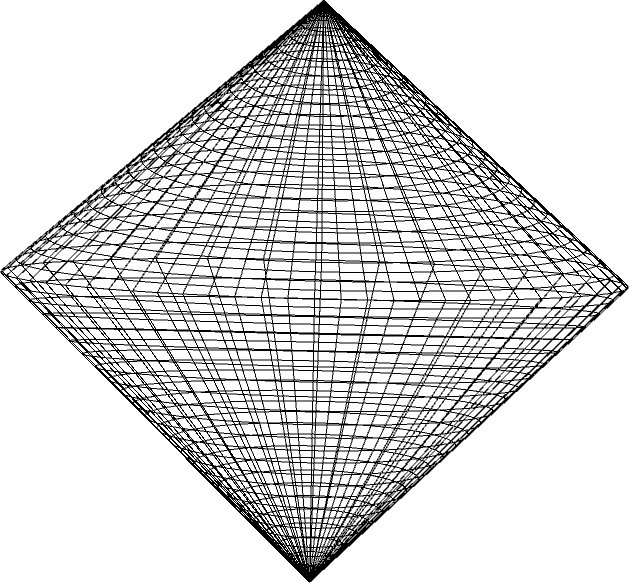}&
\includegraphics[width=0.28\columnwidth]{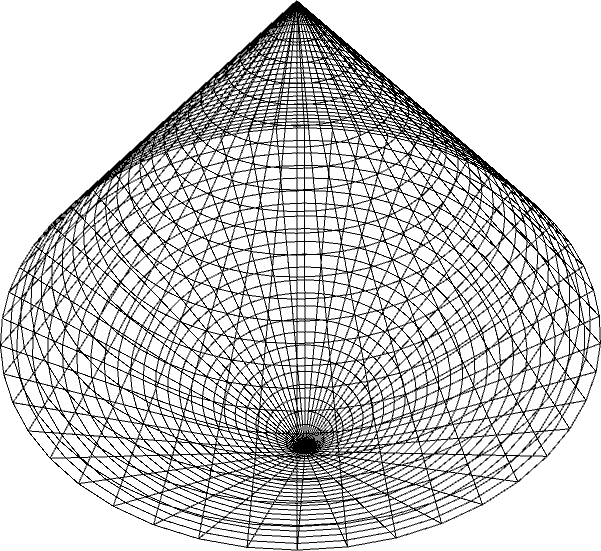}&
\includegraphics[width=0.28\columnwidth]{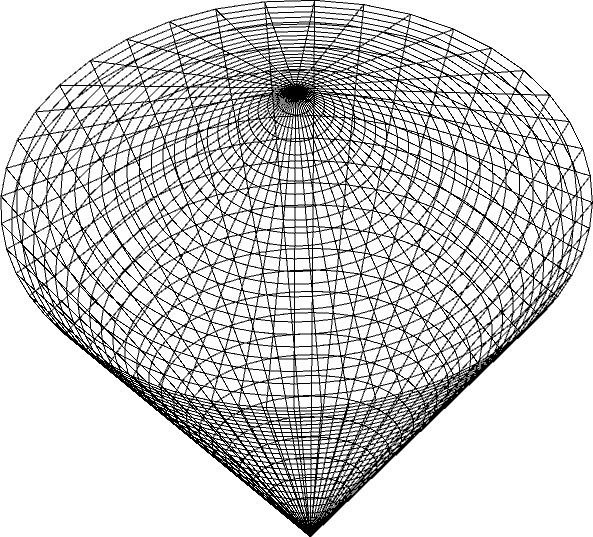}
\end{tabular}

\caption{The parameter space of $\barVPM(2)$ can be viewed as a double Lorentz cone in $\bbR^3$: Three views of the double Lorentz cone.}%
\label{fig:GaussianVPM2}%
\end{figure}

Let us mention that the VPM-Lorentz bicone analogy cannot be extended in general dimension $n$~\cite{fawzi2019representing}.

\subsection{Hilbert geometry}\label{sec:HG}

Hilbert geometry~\cite{Hilbert-1895,goldman2022geometric} is a geometry defined on any open bounded convex domain $C$ of a Hilbert space.
It defines a metric distance $d_H^C$
and ensures that the straight lines are geodesics, with uniqueness of geodesics depending on the smoothness of the boundary $\partial C$ of 
the domain~\cite{nielsen2017balls}.
When clear from context, we shall write $d_H$ instead of $d_H^C$ for sake of brevity.

Birkhoff geometry~\cite{birkhoff1957extensions,MetricProjectiveHilbert-2014} is a geometry defined on any open regular cone $K$ (i.e., convex pointed cone). 
The cone defines a partial order which allows one to define a Birkhoff ``distance'' $d_B$ which is a projective distance as it conceptually measures a distance between rays $r$ and $r'$ of the cone. 
The distance is termed projective because it satisfies the following identity: $d_B(\lambda r,\lambda' r')=d_B(r,r')$ for any $\lambda>0$ and $\lambda'>0$.
In other words, Birkhoff distance is a distance on equivalence classes $[r]$ and $[r']$ of points $r$ and $r'$ where $r \sim \lambda r$ for any $\lambda>0$. 
The Birkhoff distance is also called Hilbert projective distance or Hilbert projective metric  in the literature~\cite{nussbaum1988hilbert}. 
These two metric Hilbert geometry and projective Birkhoff geometry are related to each other when we consider for the Hilbert domain $C=K\cap H$, 
i.e., a cone slice obtained by the restriction of $K$ to any given subspace $H$.

Hilbert geometry for many different open convex domains have been studied:
For example, the Hilbert geometry of the unit disk yields Klein model of hyperbolic geometry~\cite{nielsen2020siegel}.
The Hilbert geometry of simplices have been studied in~\cite{delaharpeHibertsMetricSimplices1993} (and more generally for polytopes in~\cite{lemmensIsometriesPolyhedralHilbert2009}) with applications in machine learning in~\cite{nielsen2018clustering,nielsen2023non,vaneceksupport}. 

Birkhoff projective geometry of the symmetric positive-definite matrices (SPD) found many applications due to the fact that it exhibits a contraction property~\cite{birkhoff1957extensions,chen2021stochastic} and only requires   the extreme eigenvalues of matrices to compute the Birkhoff distance~\cite{mostajeran2024differential}.
In comparison, the affine-invariant Riemannian metric distance~\cite{thanwerdas2023n} (AIRM) requires the full spectral decomposition.
The Hilbert geometry of the space of correlation matrices, whose domains are called elliptopes~\cite{tropp2018simplicial}, was studied in~\cite{nielsen2018clustering}.
Recently, Hilbert geometry gained interests in computational geometry due to its mathematical and computational tractability of their Voronoi diagrams~\cite{DBLP:conf/compgeom/GezalyanM23} and dual Delaunay triangulations~\cite{gezalyan2024delaunay}, or smallest enclosing balls~\cite{nielsen2018clustering,banerjee2024heine}, etc.  

\subsection{Contributions and paper outline}\label{sec:contrib}

In this paper, we consider the Hilbert geometry for the  covariance-precision bicone domain.

The paper is organized as follows: 
In Section~\ref{sec:BHG}, we concisely describe the Birkhoff projective geometry of regular open cones and the Hilbert geometry defined on open bounded convex domains.
The Hilbert geometry of the open variance-precision manifold is studied in \S\ref{sec:HGVPM} with the following results:
\begin{itemize}

\item an equivalent definition of the variance-precision model of James~\cite{VarianceInfoMfd-1973} is proven in Definition~\ref{def:VPM}:
\begin{eqnarray*}
\barVPM(n) &\eqdef&\{X \in \Sym(n) : 0 \leq \lambda_i(X) \leq 1, \forall i\in\{1,\ldots,n\}\},\quad\mbox{\protect\cite{VarianceInfoMfd-1973}}\\
&=& \{X \in \Sym(n) : 0 \preceq X \preceq I\}\quad(\mbox{Remark~\ref{rk:vpmequiv}}).
\end{eqnarray*}

\item the $\VPM$ is proven open convex and bounded in Proposition~\ref{prop:vpm-open-convex},

\item the automorphisms of the VPM bicones are reported in Proposition~\ref{prop:vpm-invariance},

\item a closed-form formula for the Hilbert distance in the VPM is obtained in Theorem~\ref{def:hilbert-dist}.

\end{itemize}

In Section~\ref{sec:invariance}, we report two invariance properties of the Hilbert VPM distance: 
Namely, the invariance under transformations $X \mapsto I - X$ in Proposition~\ref{prop:ImX} and 
the invariance under orthonormal matrix conjugation 
$X \mapsto U^\top\, X\, U$ in Proposition~\ref{prop:invariance-conjugation}.
The following Section~\ref{sec:char} proves that these two isometric transformations characterized all the VPM isometries (Theorem~\ref{thm:isochar}).
Finally, we hint at some future perspectives in~\S\ref{sec:concl}.

\section{Birkhoff and Hilbert geometries}\label{sec:BHG}

\subsection{Cones and partial orders}

\begin{definition}[Cone order]
Given any vector space $V$, if $K \subseteq V$ satisfies the following conditions:
\begin{itemize}
    \item $K$ is an open cone, i.e. for any $v \in K$ and any $c \in \bbRpos$, we have $cv \in K$
    \item $K$ is convex, i.e. for any $v, w \in K$ and any $c \in [0, 1]$, we have $cv + (1-c)w \in K$
    \item $K$ is pointed, i.e. $\overline{K} \cap (-\overline{K}) = \{0\}$ 
\end{itemize}
then $K$ defines a partial order $\prec_{K}$, or simply $\prec$, on elements of $V$, by $v \prec w$ if and only if $w - v \in K$. We write $v \preceq w$ if $w - v \in \overline{K}$.
\end{definition}

An example is the Loewner order $\preceq$ on the set of symmetric matrices $\Sym(n,\bbR)$, arising from the convex open cone of symmetric positive definite matrices $\PD(n) \subset \Sym(n)$.

\begin{definition}[Loewner order]\label{def:loewner}
We define a partial order on the space of symmetric matrices $A \preceq B$ if and only if $0 \preceq B - A$.
\end{definition}

This order can also be characterised in terms of eigenvalues, as below.

\begin{proposition}[Loewner order implies eigenvalues ordering]\label{prop:loewner-eigenvalues}
    Given two symmetric matrices $A, B$ of $V=\Sym(n)$, if $A \preceq B$, then $\lambda_i(A) \leq \lambda_i(B)$, where $\lambda_i(X)$ denotes the $i$-th smallest eigenvalue of $X$ (including multiplicities). The converse holds if the ordering is true elementwise in a common eigenbase.
\end{proposition}

\begin{proof}
    For the forward direction, following~\cite{mathoverflowminmax}, assume that $A \preceq B$, that is, $0 \preceq B - A$. Then from the min-max theorem, we have:
    \[
        \lambda_i(A) = \min_{U \subset \mathbb{R}^n; \dim(U) = i}\,\, \max_{x \in U} \frac{v^\top Av}{\|v\|^2_2} \leq \min_{U \subset \mathbb{R}^n; \dim(U) = i}\,\, \max_{x \in U} \frac{v^\top Bv}{\|v\|^2_2}  = \lambda_i(B)
    \]
    
    In the other direction, we assume $A = Q^\top\, \diag(\alpha_1,\ldots, \alpha_n)\, Q$ and $B = Q^\top\, \diag(\beta_1,\ldots, \beta_n)\, Q$ in the common eigenbasis given by $Q$. If $\alpha_i \leq \beta_i$, then we have $(B - A)q_i = (\beta_i - \alpha_i)$, proving the converse.
\end{proof}

\subsection{Birkhoff projective distance}

\begin{definition}[Birkhoff projective distance]
    Given a vector space $V$ and an open convex pointed cone $K \subseteq V$ defining an order $\preceq$ on $V$, we define the Birkhoff distance $d_B$ on $K$ as
		 follows:
    \[
    d_B(v, w) = \log \frac{M(v, w)}{m(v, w)}
    \]
    where: \[
    M(v, w) = \inf_{\lambda > 0}\{v \preceq \lambda w\}
    \qquad m(v, w) = \sup_{\mu > 0}\{\mu w \preceq v\}
    \]
\end{definition}

\begin{example}
Consider the SPD cone $\PD(n)$.
Then the Birkhoff projective distance~\cite{nielsen2018clustering,chen2021stochastic} is
$$
d_B(P, Q)=\log\frac{\lambda_{\max}(PQ^{-1})}{\lambda_{\min}(PQ^{-1})}.
$$
\end{example}

\subsection{Hilbert metric distance}

\begin{definition}[Hilbert distance function]\label{def:hilbert-dist}
    Given an open convex set $X$ in a normed vector space $(V, \|\cdot\|)$ with the boundary $\partial X$, we define the Hilbert distance function $\dh(x, y)$ on elements $x, y \in X$ by the following formula:
    \[
        \dh(x, y) = \log \frac{\|x - y'\|}{\|x' - x\|}\frac{\|x' - y\|}{\|y - y'\|}
    \]
    where $x', x, y, y'$ are four points lying in this order on the line $l_{xy}$, such that $\{x', y'\} = \partial X \cap l_{xy}$.
\end{definition}

\begin{remark}
    Hilbert distance in~\Cref{def:hilbert-dist} is independent of the norm on $V$, by the equivalence of one-dimensional norms to the absolute value $|\cdot|$ on $l_{xy}$.
\end{remark}

\begin{example}
Consider $X=(0,1)$ the open unit interval of $V=\bbR$. 
Then the Hilbert distance is 
$$
d_H(x,y)=\left|\log \frac{(1-x)y}{x(1-y)} \right|.
$$
\end{example}

\subsection{Relationship between Hilbert and Birkhoff distances}

\begin{definition}[Projective space over a pointed cone]\label{prop:birkoff-hilbert-cone}
    For a pointed cone $K$ in a vector space $V$, we define the projective space over $K$ to be the space of half-lines in $K$, that is, $\mathbf{P}(K) = K / \sim$ where $x \sim y$ for $x, y \in K$ if and only if $x = \lambda y$ for some $\lambda \in \mathbb{R}_{>0}$.
\end{definition}

\begin{definition}[Open pointed cone induced by a convex set]\label{prop:birkoff-hilbert}
    Given a convex set $C \subseteq V$, we define the (open, pointed) cone over $C$ as: \[
    K(C) = \{(\lambda x, \lambda) \st x \in C, \lambda \in \bbRpos\}
    \]
\end{definition}

\begin{proposition}[{Birkhoff characterisation of the Hilbert metric,~\cite[Theorem 2.2]{lemmensbirkhoff}}]\label{prop:birkoff-hilbert-match}
    For a vector space $V$ and a open bounded convex set $C \subseteq V$, the space $\mathbf{P}(K(C))$ equipped with the Birkhoff metric $d_B$ is isometric to the space $C$ equipped with the Hilbert metric $\dh$.
\end{proposition}

\begin{example}
Consider the positive orthant cone $\bbR_{>0}^n$.
Then the Birkhoff distance is
$$
d_B(p,q)=\log \max_{i,j} \frac{p_i q_j}{q_i p_j},
$$
and the Hilbert distance on the standard simplex~\cite{nielsen2018clustering}  is
$$
d_H(p,q)=\log \frac{\max_i \frac{p_i}{q_i}}{\min_j\frac{p_j}{q_j}}.
$$
The standard simplex is interpreted as a slanted slice of the positive orthant cone.
\end{example}

\section{Hilbert geometry of the Variance-Precision Model}\label{sec:HGVPM}

\begin{definition}[Variance-precision manifold $\VPM(n)$]\label{def:VPM}
    We define the variance-precision manifold of dimension $n$ as the space:
    \[
    \VPM(n) \eqdef \{X \in \Sym(n) : 0 \prec X \prec I\}.
    \]
\end{definition}

\begin{remark}[Variance-precision model $\overline{\VPM(n)}$]\label{rk:vpmequiv}
    \cite{VarianceInfoMfd-1973} defines the variance information manifold of dimension $n$ as the space of real symmetric matrices $n \times n$, such that for all eigenvalues $\lambda_i, 1 \leq i \leq n$:
    \[
    0 \leq \lambda_i \leq 1
    \]
    This corresponds to the closure $\overline{\VPM(n)} = \{X \in \Sym(n) : 0 \preceq X \preceq I\}$ in $\Sym(n)$ in our~\Cref{def:VPM}, by the equivalence of Loewner order to eigenvalue ordering, as in~\Cref{prop:loewner-eigenvalues}.
\end{remark}

\begin{proposition}\label{prop:vpm-open-convex}
$\VPM(n) \subset \Sym(n) \simeq \mathbb{R}^{n(n+1)/2}$ is open in $\Sym(n)$, convex and bounded.
\end{proposition}
\begin{proof}
$\PD(n)$ is convex and open in $\Sym(n)$, and $\VPM(n) = \PD(n) \cap (I - \PD(n))$, i.e., an intersection of two open convex sets is open and convex. Moreover, all norms in a finite-dimensional vector space are equivalent, which means we only have to show boundedness in one norm. This means that $0 \preceq A \preceq I$ implies that:
\[
0 \preceq A^2 \preceq A \preceq I
\]
and therefore:
\[
\|A\|_F^2 = \tr(A^\top A) \leq \tr(A) \leq n.
\]
\end{proof}

\begin{proposition}\label{prop:vpm-invariance}
$\VPM(n)$ is invariant under the mapping $X \mapsto I - X$ and conjugation with orthonormal matrices $X \mapsto U^TXU$ for $U \in O(n)$.
\end{proposition}
\begin{proof}
    We use the characterisation of the Loewner order in terms of eigenvalues, as in~\Cref{prop:loewner-eigenvalues}. For the first part, observe that for each eigenvalue $\lambda$ of $X$, the inequality $0 < \lambda <1$ implies $0 < 1 - \lambda < 1$. For the second part, spectrum is invariant under conjugation with $O(n)$.
\end{proof}

Let us now compute the Hilbert distance on the interior of the Variance-Precision manifold (\Cref{def:VPM}).

\begin{theorem}[Hilbert distance on $\VPM(n)$]\label{prop:hilbert-vpm-formula}
    Given two matrices $A, B \in \VPM(n)$, 
    \[
    \dh(A, B) = \log \frac{\max(\lambda_{\max}, \mu_{\max})}{\min(\lambda_{\min}, \mu_{\min})}
    \]
    where 
    \[
        \lambda_{\min} = \lambda_{\min}(B^{-1}A), \qquad \lambda_{\max} = \lambda_{\max}(B^{-1}A),
    \]
    are the minimal and maximal eigenvalues of the $A^{-1}B$ matrix, and
    \[
        \mu_{\min} = \lambda_{\min}((I - B)^{-1}(I - A)), \qquad 
        \mu_{\max} = \lambda_{\max}((I - B)^{-1}(I - A)),
    \]
    are the minimal and maximal eigenvalues of the $(I - B)^{-1}(I - A)$ matrix.
\end{theorem}

\begin{proof}

$\VPM(n)$ is an open bounded convex set (\Cref{prop:vpm-open-convex}), therefore $\VPM(n)$ can be equipped with the Hilbert distance.
 To prove the theorem, we employ Birkhoff's characterisation of the Hilbert metric, using the Loewner order $\preceq$ (see~\Cref{def:loewner}).

We consider the affine map $(\hat{\cdot}) :\VPM(n) \to \PD(n) \times \PD(n)$ given by $\hat{A} = (A, I - A)$. The image of this map is a convex set:\[
C = \{(X, Y) \in \PD(n) \times \PD(n) : X + Y = I\}
\]
On the set $C$, we build the cone \[
K(C) = \{(\lambda X, \lambda Y, \lambda): (X, Y) \in C, \lambda\in\bbRpos\} \subset \Sym(n) \times \Sym(n) \times \bbRpos
\]
By~\Cref{prop:birkoff-hilbert}, we have the equality:
\[
d_H(A, B) = d_H(\hat{A}, \hat{B}) = d_B((\hat{A}, 1), (\hat{B}, 1)) =\log \frac{M\left((\hat{A}, 1), (\hat{B}, 1)\right)}{m\left((\hat{A}, 1), (\hat{B}, 1)\right)}
\]
where $d_H$ is the Hilbert metric on $\VPM(n)$ equal to the Hilbert metric on its affine image $C$, and $d_B$ is the Birkhoff metric on $\mathbf{P}(K(C))$, which we obtain here by considering the projective slice $[\hat{X} : 1]$. We first compute the numerator $M$:
\begin{align*}
M\left((\hat{A}, 1), (\hat{B}, 1)\right) &= \inf\left\{\lambda > 0 : (\hat{A}, 1) \preceq_{K(C)} \lambda (\hat{B}, 1)\right\},\\
&= \inf\left\{\lambda > 0 : \lambda (\hat{B}, 1) - (\hat{A}, 1) \in K(C)\right\}, \\
&= \inf\left\{\lambda > 0 : \exists_{Z \in \VPM(n)} \exists_{\theta \in \bbRpos} \text{ s.t. } \lambda (\hat{B}, 1) - (\hat{A}, 1) = (\theta \hat{Z}, \theta) \right\}.
\end{align*}
The last equality can only happen if $\theta = \lambda - 1$, and therefore $\lambda\hat{B} - \hat{A} = (\lambda - 1)\hat{Z}$, meaning:
\[
\lambda B - A = (\lambda - 1)Z, \qquad \lambda(I - B) - (I - A) = (\lambda - 1)(I - Z),
\]
where the second equality is equivalent to the first one. Since $0 \prec Z \prec I$, we arrive at two inequalities:
\[
A \prec \lambda B, \qquad (I - A) \prec \lambda (I - B)
\]
Performing an analogous computation for $m\left((\hat{A}, 1), (\hat{B}, 1)\right)$, we get:
\begin{align*}
M(\hat{A},\hat{B}) &= \inf_{\lambda > 0}\{ A\preceq \lambda B \text{ and } (I - A) \preceq \lambda (I - B)\}, \\
m(\hat{A},\hat{B}) &= \sup_{\mu > 0}\{\mu B\preceq A \text{ and } \mu(I-B)\preceq I-A\}.
\end{align*}
Still focusing on $M(\hat{A}, \hat{B})$, if $A \preceq \lambda B$, then $A^{-1/2}AA^{-1/2} \preceq \lambda A^{-1/2}BA^{-1/2}$. The converse in \Cref{prop:loewner-eigenvalues} proves that this holds if and only $\frac{1}{\lambda} \leq \lambda_{\min}(A^{-1}B) = \lambda_{\min}(A^{-1/2}BA^{-1/2})$, and therefore: 
\[
\lambda \geq \frac{1}{\lambda_{\min}(A^{-1}B)} = \lambda_{\max}(B^{-1}A)\]
Similarly, applied to $(I-A)\preceq \lambda (I - B)$, we get $\lambda \geq \lambda_{\max}((I - B)^{-1}(I - A))$. Taking the infimum over $\lambda$, we obtain:
\[
M(\hat{A}, \hat{B}) = \max\left(\lambda_{\max}(B^{-1}A), \lambda_{\max}((I - B)^{-1}(I - A))\right).
\]
The case of $m(\hat{A}, \hat{B})$ follows in the same way. If $\mu I \leq B^{-1}A$ and $\mu I \leq (I - B)^{-1}(I - A)$, then:
\[
\mu \leq \lambda_{\min}(B^{-1}A), \qquad 
\mu \leq \lambda_{\min}((I - B)^{-1}(I - A)),
\]
and so taking the infimum:
\[
m(\hat{A}, \hat{B}) = \min\left(\lambda_{\min}(B^{-1}A), \lambda_{\min}((I - B)^{-1}(I - A))\right).
\]
Substituting it in the original formula, we get the claimed formula:
\[
d_H(A, B) = \log \frac{\max(\lambda_{\max}, \mu_{\max})}{\min(\lambda_{\min}, \mu_{\min})},
\]
for $\lambda_{\min, \max}$ the minimal and maximal eigenvalues of $B^{-1}A$, and $\mu_{\min, \max}$ the minimal and maximal eigenvalues of $(I- B)^{-1}(I - A)$.
\end{proof}

\begin{remark}
    Because of the logarithm of ratios, in~\Cref{prop:hilbert-vpm-formula}, the Hilbert distance on $\VPM(n)$ can be equivalently written using matrices $A^{-1}B$ and $(I - A)^{-1}(I - B)$ in place of $B^{-1}A$ and $(I - B)^{-1}(I- A)$ respectively.
		The transformation $\Mob(A,B)=(I - A)^{-1}(I - B)$ is called a matrix M\"obius transformation.
\end{remark}

\section{Invariance properties of the Hilbert VPM distance}\label{sec:invariance}

Let us first state the following two propositions:

\begin{proposition}[$\iota$ map]{{\cite{VarianceInfoMfd-1973}}}
    The variance-precision manifold is diffeomorphic to the set of symmetric positive-definite matrices of dimension $n$ via the map $\iota(X) = X(I + X)^{-1}$ and its inverse $\iota^{-1}(A) = A(I - A)^{-1}$.
\end{proposition}

Those mappings can be interpreted as matrix M\"obius transformations.

\begin{proposition}
    The derivative of the map $\iota$ and its inverse $\iota^{-1}$ are given by:
    \[
    d\iota_X(H) = (I + X)^{-1}H(I + X)^{-1} \qquad d\iota^{-1}_A(H) = (I - A)^{-1}H(I - A)^{-1}
    \]
\end{proposition}

The formula of the Hilbert VPM distance from~\Cref{prop:hilbert-vpm-formula} immediately implies the following.
\begin{proposition}[Identity-complement invariance]\label{prop:ImX}
The map $X \mapsto I - X$ is an isometry of $\VPM(n)$.
\end{proposition}

We observe that the map $X \mapsto X^{-1}$ is an isometry for the  affine invariant Riemannian metric (AIRM)~\cite{pennec2006riemannian} (and for the log-Euclidean metric~\cite{arsigny2006log}).

\begin{remark}
The invariance of the Hilbert metric on $\VPM(n)$ under map $X \mapsto I -X$ translates to the invariance of the standard Riemannian metric on $\PD(n)$ under matrix inverse $X \mapsto X^{-1}$.
\end{remark}

\begin{proof}
We compute:
\begin{align*}
(\iota^{-1} \circ (I - \cdot) \circ \iota)(X) &= \iota^{-1}(I - \iota(X)) \\
&= \iota^{-1}(I - X(I + X)^{-1}) \\
&= (I - X(I + X)^{-1})(I - I + X(I + X)^{-1})^{-1} \\
&= (I - X(I + X)^{-1})(I + X)X^{-1} \\
&= ((I + X)X^{-1} - I) \\
&= X^{-1}(I + X - X) = X^{-1}
\end{align*}
\end{proof}

\begin{proposition}\label{prop:invariance-conjugation}
Conjugation under orthonormal matrix $U$, that is, a map $X \mapsto U^\top XU$ is an isometry in the Hilbert metric on $\VPM(n)$.
\end{proposition}

\begin{proof}
    First, we note that eigenvalues are invariant under conjugation, and conjugation with an orthonormal matrix $U \in O(n)$  is a congruence, which preserves symmetry, therefore conjugation $X \mapsto U^TXU = U^{-1}XU$ is an automorphism of $\VPM(n)$. Given two matrices $X, Y \in \VPM(n)$ and an orthonormal matrix $U$, we write $A = U^{-1}XU$ and $B = U^{-1}YU$ and calculate:
    \[
    \dh(A, B) = \log \frac{\max(\lambda_{\max}(A^{-1}B), \lambda_{\max}((I - A)^{-1}(I - B)))}{\min(\lambda_{\min}(A^{-1}B), \lambda_{\min}((I - A)^{-1}(I - B)))}
    \]
    But then:
    \[
    A^{-1}B = (U^{-1}XU)^{-1} (U^{-1}YU) = (U^{-1}X^{-1}U) (U^{-1}YU) = U^{-1}(X^{-1}Y)U
    \]
    and again, the eigenvalues are invariant under conjugation, so $\lambda_{\min, \max}\left(A^{-1}B\right) = \lambda_{\min, \max}\left( X^{1-}Y\right)$.
    We also observe that:
    \begin{align*}
    (I - A)^{-1}(I - B) &= (I - Q^{-1}XQ)^{-1}(I - Q^{-1}YQ) 
    \\
    &= (Q^{-1}(I - X)Q)^{-1}(Q^{-1}(I - Y)Q) 
    \\
    &= Q^{-1}(I - X)^{-1}(I - Y)Q    
    \end{align*}
    and therefore the same argument applies here. Thus, the Hilbert distance is preserved.
\end{proof}

Similarly, the conjugation with the orthonormal matrix $X \mapsto U^TXU$ is an isometry for the AIRM metric (and for the log-Euclidean metric).

\begin{remark}
    The invariance of the Hilbert metric on $\VPM(n)$ under conjugation with orthonormal matrix $X \mapsto U^TXU$ translates into the invariance of the standard Riemannian metric on $\PD(n)$ under conjugation $X \mapsto U^TXU$.
\end{remark}
\begin{proof}
We compute:
\begin{align*}
U^T\iota(X)U &= U^TX(I + X)^{-1}U \\
&= (U^TX U)\left(U^T(I + X)^{-1}U\right) \\
&= (U^TX U)\left((U^T(I + X)U)^{-1}\right) \\
& = (U^TX U)\left((I + U^TXU)^{-1}\right) \\
& =\iota(U^TXU)
\end{align*}
and similarly for $U^T\iota^{-1}(X)U = \iota^{-1}(U^TXU)$. Thus:
\[
(\iota^{-1} \circ (U^T\cdot U) \circ \iota)(X) = U^TXU
\]
\end{proof}


Next, we show that these two invariant transformations are the only ones by characterizing the isometries of the VPM when $n\geq 2$.

\section{Characterisation of VPM isometries}\label{sec:char}

The argument is based on the simplification from isometries of Hilbert metric to collineations of the underlying cone (for the Birkhoff characterisation). We assume $n \geq 2$ everywhere below. We will denote the group of isometries of a metric space $X$ by $\Isom(X)$.

\begin{definition}
    Given a real vector space $V$, three points $x, y, z \in X$ are called collinear if there exist real constants $\lambda, \lambda'$ such that: \[
    (x - y) = \lambda(x - z) = \lambda'(y - z)
    \]
\end{definition}

\begin{definition}
    Given two vector spaces $V, W$ and subsets $X \subseteq V$ and $Y \subseteq W$,  an invertible transformation $f: X \to Y$ is called a collineation if for every three collinear points $x, y, z \in X$, their images $f(x), f(y), f(z)$ are also collinear. The group of collineations of $W$ is denoted $\Coll(W)$.
\end{definition}

\begin{proposition}[{\cite{goldman2022geometric}}]\label{prop:collineations-are-isometries}
    For an open convex bounded set $C \subseteq \mathbb{R}^n$ equipped with the Hilbert geometry $\dh$, every collineation of $C$ is also an isometry of $C$. That is, $\Coll(C) \subseteq \Isom(C)$.
\end{proposition}

The converse of~\Cref{prop:collineations-are-isometries} is not true in general. But we establish it for $\VPM(n)$ in particular, using the theory developed in~\cite{walshGaugereversingMapsCones2017}.

\subsection{Isometries of VPMs are collineations}

\begin{definition}[Lorentzian cone]
    A cone $C \subseteq \mathbb{R}^n$ is called Lorentzian if it's of a form $C = \left\{(x_1, \ldots x_n) : x_1 > 0,\ x_1^2 - \sum_{i > 1}x_i^2 > 0 \right\}$
\end{definition}

\begin{definition}[Proper cone]
    An open convex cone $C$ in a linear space $V$ is called proper if it does not contain any subspace of $V$, or in other words if we have both $x \in C$ and $-x \in C$ iff $x = 0$.
\end{definition}

\begin{definition}[Dual cone]
    For a cone $C$ in a linear space $V$ equipped with an inner product $\langle \cdot, \cdot \rangle$, its dual cone is defined to be \[
    C^* = \{y \in V : \langle x, y \rangle > 0 \text{ for all } x \in C\}
    \]
\end{definition}

\begin{definition}[Homogeneous cone]
    A cone $C \subseteq V$ is said to be homogeneous, if for every two points $x, y \in C$ there exists a linear automorphism $F$ of $V$, such that $F$ restricts to an automorphism of $C$, and $F(x) = y$.
\end{definition}

\begin{definition}[Symmetric cone]
    An open convex proper cone $C$ is said to be symmetric if it is homogeneous, and equal to its dual $C = C^*$.
\end{definition}

\begin{theorem}[{\cite{walshGaugereversingMapsCones2017}}, Corollary 1.4]\label{thm:iso-coll}
    Given a cone $C$ and $D = \mathbf{P}(C)$ its projective space, we have that $\Isom(D)$ is a normal subgroup of $\Coll(D)$ if and only if $C$ is symmetric and non-Lorentzian. Otherwise, $\Isom(D) = \Coll(D)$.
\end{theorem}

We apply this theorem to our space.

\begin{definition}[Cone $C_n$ over $\VPM(n)$]\label{def:vpm-cone}
    The cone over $\VPM(n)$ is given by:
    \[
    C_n = K(\VPM(n)) = \left\{(tX, t) : X \in \VPM(n),\ t \in \bbRpos\right\} = \left\{(Y, t) : \frac{1}{t}Y \in \VPM(n), \ t \in \bbRpos\right\}
    \]
    for which $\VPM(n) = \mathbf{P}(C_n)$ is the projective space - in the projective coordinates we have $\VPM(n) \simeq [X: 1]$.
\end{definition}

\begin{proposition}\label{prop:VPM-not-symmetric}
    The cone $C_n$ for $n > 1$ is not symmetric.
\end{proposition}
\begin{proof}
    We form the dual cone:\[
    C_n^* = \{(Y, s) \in \Sym(n) \times \mathbb{R} : \langle (X, t), (Y, s) \rangle > 0 \text{ for all } (X, t) \in C \}
    \]
    where the inner product is induced on the product space as:
    \[
    \langle (X, t), (Y, s) \rangle = \langle X, Y \rangle + ts = \text{tr}(XY) + ts
    \]
    We denote $Z = \frac{1}{t}X$, where $Z \in \VPM(n)$ and $t > 0$ are now arbitrary, and write:
    \[
    t(\text{tr}(ZY) + s) > 0
    \]
    which is equivalent to $\text{tr}(ZY) + s> $. Because $\VPM(n)$ is invariant under conjugation with orthonormal matrices as~\Cref{prop:invariance-conjugation} showed, we can diagonalize $Y = Q^\top\text{diag}(\lambda_1, \ldots, \lambda_n)Q$ and compute the infimum:
    \begin{align*}
    \inf_{Z \in \VPM(n)} \text{tr}(ZY)
    &= \inf_{Z \in \VPM(n)} \sum_{i = 1}^n \lambda_i q_i^TZq_i \\
    &\geq \inf_{Z \in \VPM(n)} \sum_{\lambda_i < 0} \lambda_i q_i^TZq_i \\
    &\geq \sum_{\lambda_i < 0} \lambda_i  = -\text{tr}(Y_{-})   
    \end{align*}
    and the infimum is achieved by letting $Z \to Q^T\text{diag}({\mathbbm{1}_{\lambda_1 < 0}, \ldots, \mathbbm{1}_{\lambda_n < 0}})Q$ (which is attained in the closure of $\VPM(n)$ in $\Sym(n)$). This means that the sufficient and necessary condition defining the dual cone is: \[
    C_n^* = \{(Y, s) \in \Sym(n)\times \mathbb{R} : s > \text{tr}(Y_{-})\}
    \]
    which is clearly different from $\VPM(n)$, since it does not even depend on the positive eigenvalues of $Y$.
\end{proof}

Therefore, from~\Cref{prop:VPM-not-symmetric} and~\Cref{thm:iso-coll} we derive the desired characterisation.
\begin{corollary}\label{prop:isometries-are-collineations}
We have $\Isom(\VPM(n)) = \Coll(\VPM(n))$ for all $n \geq 2$.
\end{corollary}

\subsection{Classification of collineations of VPM}

Let us classify the collineations of $\VPM(n)$. To do that, we use the fact that $\VPM(n)$ can be seen as the projective space of its cone $\mathbf{P}(C_n)$, as in in~\Cref{def:vpm-cone}.

\begin{theorem}[Fundamental theorem of projective geometry]\label{thm:fundamental}
     Any bijective collineation $\hat{f}: \mathbb{RP}^n \to \mathbb{RP}^n$ is induced by some linear isomorphism $f: \mathbb{R}^{n+1} \to \mathbb{R}^{n+1}$, such that $\hat{f} = [f]$, where $[f][x] = [fx]$.
 \end{theorem}
 
\begin{definition}[$\PGL(n)$]
    By $\PGL(n)$ we denote the projective general linear group of dimension $n$, that is, the group $\GL(\mathbb{R}, n) / \GL(\mathbb{R}, 1)$. That is, $F \sim G$ iff $F = aG$ for some real constant $a \neq 0$.
\end{definition}

\begin{proposition}[{\cite{shiffmanSYNTHETICPROJECTIVEGEOMETRY1995}}, Lemma 4]\label{lemma:collineations-extend}
    Given an open set $U \subseteq \mathbb{RP}^n$ for $n \geq 2$, and a continuous injective map $f: U \to \mathbb{RP}^n$, such that for all projective lines $L \subseteq  \mathbb{RP}^n$, the set $f(L\cap U)$ is collinear, then there exist a collineation $\hat{f} : \mathbb{RP}^n \to \mathbb{RP}^n$ such that $\hat{f}|_U = f$.
\end{proposition}

Combining~\Cref{lemma:collineations-extend} and~\Cref{thm:fundamental} instead of classifying collineations of $\VPM(n)$, we can instead classify the linear cone isomorphisms, as we do in the next section below.

\subsubsection{Classification of the projectivisations of linear cone isomorphism}

\begin{definition}[Linear cone isomorphism]
    For a vector space $V$ and a cone $C \subseteq V$, we call a map $L : V \to V$ a linear cone automorphism for $C$, if $L$ is a linear isomorphism and $L(C) = C$.

\end{definition}

We need the following lemmas.

\begin{lemma}[{\cite{gowdaAutomorphismGroupCompletely2013a}}, Example 1]\label{lemma:UTU}
    Linear cone isomorphism $\mathcal{A}: \Sym(n) \to \Sym(n)$ for $\PD(n)$ exactly correspond to conjugation with some $U \in \GL(\mathbb{R}, n)$, that is:
    \[
    \mathcal{A}(X) = U^TXU
    \]
\end{lemma}

\begin{lemma}\label{lemma:curve-in-psd}
     Given a vector space $V$, a pointed closed cone $K$ such that $0 \in K$, and a smooth curve $\eta(t) \in K$, for $t \in (-\epsilon, \epsilon)$ and $\epsilon > 0$, if $\eta(0) = 0$, then necessarily $\eta'(0) = 0$.
\end{lemma}
\begin{proof}
    For any $t \in \bbRpos$, we have $\frac{\eta(t)}{t} \in K$, because $K$ is a cone. Since $K$ is closed, for $t > 0$, we also have \[
    \lim_{t \to 0^+}\frac{\eta(t)}{t} = \lim_{t\to 0^+}\frac{\eta(t) -\eta(0)}{t} = \eta'(0) \in K
    \]
    By a similar argument, this time taking $t < 0$, we have:
    \[
    \lim_{t \to 0^-}\frac{\eta(t)}{-t} = -\eta'(0) \in K
    \]
    But $K$ is pointed, so $K \cap (-K) = \{0\}$, and therefore $\eta'(0) = 0$.
\end{proof}

\begin{lemma}\label{lemma:affine-maps-vpm}
    Given an affine isomorphism $L: \Sym(n) \to \Sym(n)$ where $L(\VPM(n)) = \VPM(n)$, we either have:
    \[
        \begin{aligned}
            L(0) &= 0 \\
            L(I) &= I
        \end{aligned}\qquad \text{ or } \qquad 
        \begin{aligned}
            L(0) &= I \\
            L(I) &= 0
        \end{aligned}
    \]
\end{lemma}
\begin{proof}
    A general affine isomorphism $L: \Sym(n) \to \Sym(n)$ has a form $LX = \mathcal{A}X + B$ for some linear isomorphism $\mathcal{A}: \Sym(n) \to \Sym(n)$ and $B \in \Sym(n)$. Since $L$ preserves $\VPM(n)$ and is smooth, it also maps the closure $\overline{\VPM(n)}$, and thus the $\partial \VPM(n) = \overline{\VPM(n)} \backslash\VPM(n)$ to itself $L(\partial{\VPM(n)}) = \partial{\VPM(n)}$. 
     
    Let us assume that $L(X) = 0$ for some $X \in \overline{\VPM(n)}$ and $X \not\in \{0, I\}$. First, we consider a case where $X = \alpha I$ for some $0 < \alpha < 1$. Take some $0 < \beta < \min(\alpha, 1-\alpha)$, therefore $\alpha I \pm \beta I \in \VPM(n)$, and compute:
    \[
    L(\alpha I \pm \beta I) = \mathcal{A}(\alpha I \pm \beta I) + B = \left(\mathcal{A}(\alpha I) + B\right) \pm \mathcal{A}\beta I = L(\alpha I) \pm \mathcal{A}\beta I = \pm \mathcal{A}\beta I \in \VPM(n)
    \]
    But having both $\mathcal{A}I \prec 0$ and $0 \prec \mathcal{A}I$ is impossible.
    
    Next, assume there exist at least two distinct eigenvalues $\lambda_1 < \lambda_2$ of $X$, with corresponding unit eigenvectors $v_1, v_2$. Let $U_\theta$ for $\theta \in (0, 2\pi]$ be a rotation matrix which rotates the subspace $V = \text{span}(v_1, v_2)$ by angle $\theta$, that is:
    \[
    U_\theta = \exp\left(\theta Z\right) \qquad Z = v_1v_2^T - v_2v_1^T
    \]
    Since $U_\theta \in O(n)$, both $\Sym(n)$ and $\PSD(n)$ are invariant under conjugation with $U_\theta$.
    
    We thus have a smooth curve $\eta: S^1 \to \Sym(n)$ given by $\eta(\theta) = L(U_\theta^T XU_\theta)$ lying inside $\overline{\VPM(n)} \subset \PSD(n)$, such that for $\theta = 0$, we have $\eta(0) = 0$. This means that $\eta'(0) = 0$, by~\Cref{lemma:curve-in-psd}. By standard computation, we have:
    \[
    \frac{d}{d\theta}U_\theta^TXU_\theta = U_\theta^T(XZ - ZX)U_\theta
    \]
    and therefore:
    \[
    \frac{d}{d\theta} \eta(\theta) = \mathcal{A}\left(U_\theta^T(XZ - ZX)U_\theta\right)
    \]
    Thus, $\mathcal{A}(XZ - ZX) = 0$, meaning $XZ - ZX = 0$. Unfolding:
    \[
    X(v_1v_2^T - v_2v_1^T) - (v_1v_2^T - v_2v_1^T)X = \lambda_1v_1v_2^T - \lambda_2 v_2v_1^T - \lambda_2v_1v_2^T + \lambda_1 v_2v_1^T = (\lambda_1 - \lambda_2)(v_1v_2^T + v_2v_1^T)
    \]
    But if $v_1v_2^T + v_2v_1^T$ is a non-zero symmetric matrix, implying $\lambda_1 = \lambda_2$, a contradiction.

    Thus, we can only have $L(0) = 0$ or $L(I) = 0$. Let us compute:
    \[
    L(0) + L(I) = \mathcal{A}0 + B + \mathcal{A}I + B = \mathcal{A}X + B + \mathcal{A}(I - X) + B = L(X) + L(I - X)
    \]
    Since $L$ is onto $\VPM(n)$ and a diffeomorphism, it is also onto $\overline{\VPM(n)}$, so for each $Y = L(X)$, we also have $L(I - X) = C - L(X) \in \overline{\VPM(n)}$, for $C := L(0) + L(1)$. Thus, substituting $Y = 0$, we get $C \in \overline{\VPM(n)}$, i.e. $C \preceq I$, and substituting $Y = I$, we get $C - I \in \overline{\VPM(n)}$, i.e. $I \preceq C$, forcing $I = C$. Thus, $L(0) + L(I) = I$, and if $L(0) = 0$, then $L(I) = I$. Otherwise, $L(0) = I$ and $L(I) = 0$, proving the thesis.
\end{proof}

Now, our major result is the following theorem.
\begin{theorem}[Classification of projectivizations of linear cone automorphisms]
    Given a linear cone automorphism $L: \Sym(n) \times \mathbb{R} \to  \Sym(n) \times \mathbb{R}$ for $C_n$, there exists $P \in O(n)$ and $\epsilon \in \{0, 1\}$ such that:
    \[
    [L(X, t)] = [P^T((1-\epsilon)X + \epsilon(I - X))P : 1]
    \]
\end{theorem}
\begin{proof}
A linear cone automorphism $L: \Sym(n) \times \mathbb{R} \to \Sym(n) \times \mathbb{R}$ is, in general, of a form:
\[
L(X, t) = (\mathcal{A}X + tB, \langle C, X \rangle + st) =(\mathcal{A}X + tB, \text{tr}(CX) + st)
\]
for some linear $ \mathcal{A}: \Sym(n) \to \Sym(n)$ and $B, C \in \Sym(n)$ and $s \in \mathbb{R}$. 
Because $L$ maps the origin $(0, 0)$ to itself, and maps straight lines to straight lines, it has to map rays through the origin to rays through the origin. Thus, it is an isomorphism on each projective slice. That is, for each fixed $t_0 \in \bbRpos$, we get an affine isomorphism $\pi \circ L(\cdot, t_0) \circ \iota_{t_0}: \Sym(n) \to \Sym(n)$, where $\pi(X, t) = \frac{1}{t}X$ is the projectivization and $\iota_{t_0}(X) = (X, t_0)$ is an inclusion. After applying~\Cref{lemma:affine-maps-vpm}, we get that either $L(I, t) = (I, t')$ and $L(0, t) = L(0, t'')$ or $L(I, t) = (0, t')$ and $L(0, t) = (I, t'')$ for some $t', t'' \in \bbRpos$. By possibly precomposing with a map $(X, t) \mapsto (It - X, t)$, we can ensure that $L(0, \bbRpos) = (0, \bbRpos)$ and $L(I, \bbRpos) = (I, \bbRpos)$. Setting $X = 0$, we have:
\[
L(0, t) = (\mathcal{A}0 + tB, \text{tr}(C0) + st) = (tB, st)
\]
which means $B = 0$. Therefore, we can write simpler formula for the inverse:
\[
L^{-1}(X, t) = \left(\mathcal{A}^{-1}X, -\frac{\text{tr}(C\mathcal{A}^{-1}X)}{s} + \frac{t}{s}\right)
\]
Next, substituting $X = I$, and remembering that $L(I, \bbRpos) = (I, \bbRpos)$, we have:
\begin{align*}
L(I, t) &= \left(\mathcal{A}I, \text{tr}(CI) + st\right) = \left(I, -\text{tr}(C) + \frac{t}{s}\right)\\
L^{-1}(I, t) &= \left(\mathcal{A}^{-1}I, -\frac{\text{tr}(C\mathcal{A}^{-1}I)}{s} + \frac{t}{s}\right) = \left(I, -\frac{\text{tr}(C)}{s} + \frac{t}{s}\right)
\end{align*}
From the fact that $\VPM(n)$ is mapped into itself, and the fact that $L$ is a linear isomorphism and thus a diffeomorphism, we get that the closure $\overline{\VPM(n)}$, and thus the boundary $\VPM(n) \backslash \overline{\VPM(n)}$, is preserved by both $L$ and $L^{-1}$. Therefore:
\begin{align*}
I \preceq (\text{tr}(C) + s)I &\implies 1 -s\leq \tr(C) \\
I \preceq \left(-\frac{\text{tr}(C)}{s} + \frac{1}{s}\right)I &\implies 1 - s \geq \tr(C)
\end{align*}
where we also used the fact that $s \geq 0$, from the fact that $L$ preserves the positivity of rays $(0, \bbRpos)$. Therefore, we conclude that $1 - s = \text{tr}(C)$. 

Now, we note that since we're interested only in the equivalence class $[L]$, we can assume w.l.o.g. that $s = 1$:
\[
[L] = \left[\mathcal{A}X, \text{tr}(CX) + st\right] = \left[\frac{\mathcal{A}}{s}X, \text{tr}\left(\frac{C}{s}X\right) + t\right]
\]
and this combined with the previous point, sets $\text{tr}(C) = 0$. Finally, assume that $C$ has at least one negative eigenvalue $\lambda$, with the unit eigenvector $v$. Set $X = v v^T$, and compute:
\[
\tr(CX) = \text{tr}(C v v^T) = \text{tr}(\lambda v v^T) = \text{tr}(\lambda v v^T I) = \lambda \text{tr}(\lambda v^T v) =\lambda ||v||_2 = \lambda < 0
\]
Now, pick some $0 < t_0 < -\lambda$, and compute:
\[
L(t_0X, t_0) = (t_0\mathcal{A}X, \lambda - t_0)
\]
making the second component negative, which is a contradiction with the positivity of rays. This means that all eigenvalues $\lambda^C_i$ of $C$ must be non-negative, and thus:
\[
0 = \text{tr}(C) = \sum_i \lambda^C_i \implies \lambda_i^C = 0
\]
So far, the above argument proved that:
\[
[L(X, t)] = [\mathcal{A}X, t]
\]
Now, using~\Cref{lemma:UTU} and the fact that \[
L(\PD(n), \bbRpos) = (\PD(n), \bbRpos)
\]
which holds because every element of $Y \in \PD(n)$ can be scaled to a point in $\VPM(n)$ by a non-negative factor $\frac{1}{\lambda_{\max}(X) + 1}$, we have $[L(X, t)] = [P^TXP, t]$. This proves the statement of the theorem, since we have been possibly precomposing $L$ with $I - X$, and:
\[
    P^T(I - X)P = P^TP - P^TXP = I - P^TXP
\] which used the fact that $\mathcal{A}I = I$, i.e. that $P$ must be orthonormal.

\end{proof}

\begin{theorem}\label{thm:isochar}
The group of isometries of $\VPM(n)$ for $n > 1$ is generated by conjugation by orthonormal matrix $X \mapsto U^\top XU$ and inversion $X \mapsto I - X$.
\end{theorem}

\section{Discussion with some perspectives}\label{sec:concl}

We motivated this work by showing that the parameter space of the extended Gaussian family is a base-to-base closed bicone called the variance-precision model.
We then consider the Hilbert geometry for the corresponding interior of this parameter space, reported the closed-form formula of its corresponding Hilbert metric distance, and fully characterize the isometries.
We use two transformations $T_1(Q)=Q(I+Q)^{-1}$ and $T_2(Q)=(I+Q)^{-1}$ which maps the open domains, from PD to VPM, 
In the context of extended Gaussians, $T_1$ is the covariance transformation  and $T_2$ the precision transformation.
Those transformations are related by the $T_1(Q^{-1})=T_2(Q)$ and $T_2(Q^{-1})=T_1(Q)$.
For non-degenerate matrices $\Sigma=P^{-1}\Leftrightarrow P=\Sigma^{-1}$, eigenvalues of the covariance matrix $\Sigma$ and its corresponding precision matrix $P$ are reciprocal to each others: That is, it holds that
$\lambda_i(\Sigma)=\frac{1}{\lambda_{n+1-i}(P)}>0$ for $i\in\{1,\ldots,n\}$.
However, when some eigenvalues of $\Sigma$ or $P$ vanish, say $\lambda_i(\Sigma)=0$ or $\lambda_j(P)=0$, we cannot associate a corresponding $+\infty$ eigenvalue in the corresponding ``dual matrix'' (indeed, doing so will end up with an undefined matrix).
This observation explains the doubling of the PSD cone boundary enclosing the VPM when considering both covariance and precision degeneracies.

Coincidentally, the well-known affine invariant Riemannian metric (AIRM) distance~\cite{pennec2006riemannian} obtained by taking the Riemannian trace metric in $\PD$ with length element $\ds_Q^2(\dQ)=\tr\left((Q^{-1}\dQ)\right)$  at $Q\in\PD$  was calculated in James' same paper~\cite{VarianceInfoMfd-1973}. 
(James considered the Fisher metric of Wishart distributions to compute the Fisher-Rao distance.)
The closed-form formula for the AIRM distance between $Q_1,Q_2\in\PD(n)$ requires the full set of eigenvalues:
$$
\rho(Q_1,Q_2)=\sqrt{\sum_{i=1}^n\log^2 \lambda_i(Q_1Q_2^{-1})}.
$$
This Riemannian distance is invariant by inversion $\rho(Q_1^{-1},Q_2^{-1})=\rho(Q_1,Q_2)$ and  congruence transformation by invertible matrices
$\rho(A Q_1 A^\top,A Q_2 A^\top)=\rho(Q_1,Q_2)$ for all $A\in\GL(n)$.
Hilbert VPM distance requires only four extreme eigenvalues to be computed.
Let us notice that Hilbert VPM distance being only invariant under orthonormal matrix congruence transformation and not for general invertible matrix congruence transformation may be of interest in applications like in Diffusion Tensor Imaging~\cite{arsigny2006log} (DTI).

Table~\ref{tab} summarizes the features of the Hilbert VPM distance versus the AIRM distance.

\begin{table}\centering
\caption{Comparison of the AIRM vs Hilbert VPM distances.}\label{tab}

\begin{tabular}{lcc}
 & AIRM distance & Hilbert VPM distance\\ \hline
Eigenvalues & $\lambda_i(Q_1Q_2^{-1})$ & $\lambda_1(Q_1Q_2^{-1}), \lambda_n(Q_1Q_2^{-1})$\\
& & $\lambda_1(\Mob(Q_1,Q_2)), \lambda_n(\Mob(Q_1,Q_2))$\\
& & $\Mob(Q_1,Q_2)=(I-Q_2)^{-1}(I-Q_1)$\\
invariance under inversion & \checkmark & \xmark \\
 & & replaced by Identity-complement invariance \\
invariance under $\GL(n)$ congruence & \checkmark & \xmark\\
invariance under $O(n)$ congruence & \checkmark & \checkmark \\ \hline
\end{tabular}
\end{table} 

One interesting property is that we can enlarge the $\barVPM$  by defining:
$$
\barVPM_\eps=\{ -\eps\, I\preceq X \preceq (1+\eps)\,I\},\quad \eps\geq 0.
$$
See Figure~\ref{fig:GaussianVPMeps2}. Let $d_{H,\eps}$ denote the Hilbert distance induced by $\barVPM_\eps$.

\begin{figure}%
\centering
\begin{tabular}{ccc}
\includegraphics[width=0.28\columnwidth]{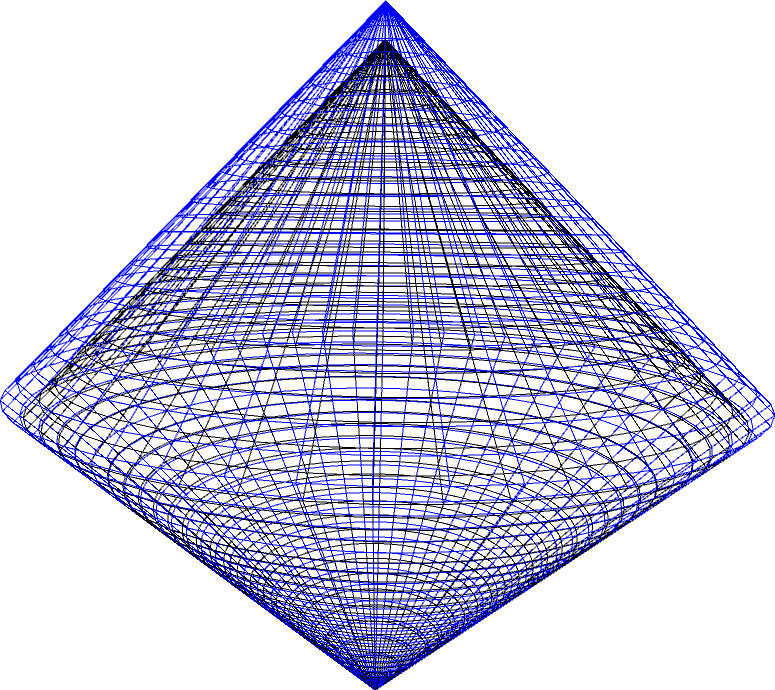}&
\includegraphics[width=0.28\columnwidth]{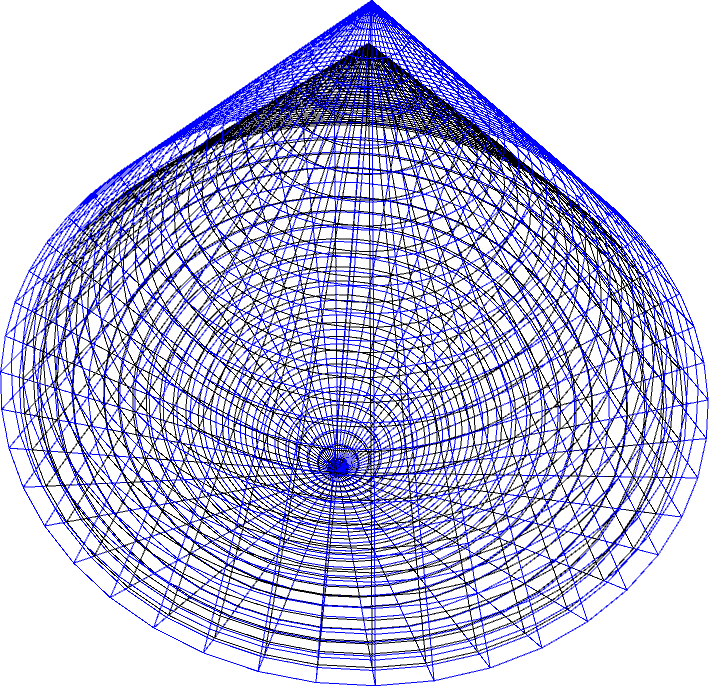}&
\includegraphics[width=0.28\columnwidth]{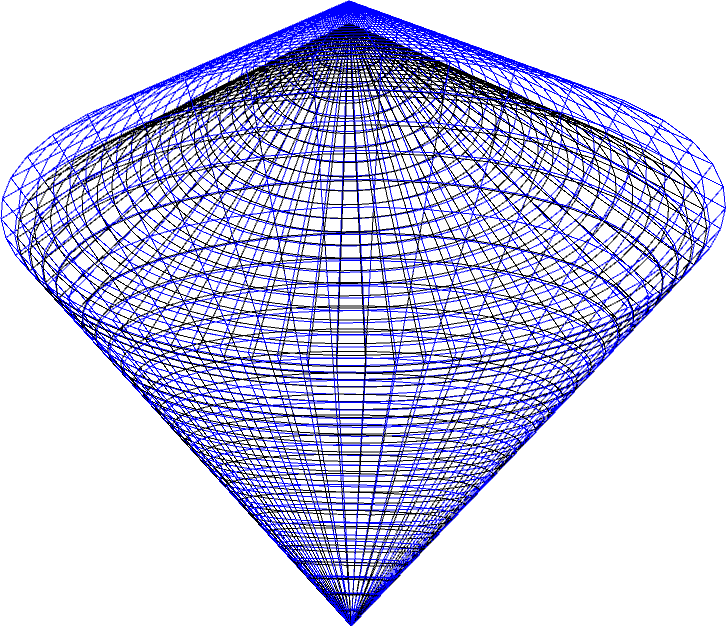}
\end{tabular}

\caption{The parameter space of $\barVPM_\eps(2)$ (viewed as a Lorentz bicone in $\bbR^3$): Three views of the boundary of the  Lorentz bicone with $\barVPM_\eps$ (in blue) encapsulating $\barVPM$ (in black).}%
\label{fig:GaussianVPMeps2}%
\end{figure}

We have the following monotonicity property Hilbert distances defined by nested $\VPM_\eps$ domains~\cite{KleinHilbertPoincare-1999,nielsen2020siegel}:
 
\begin{proposition}[Lower bounding Hilbert VPM distance]
$$
\forall \eps>0, \forall S_1,S_2\in\VPM(n), d_H(S_1,S_2) \geq d_{H,\eps}(S_1,S_2).
$$
\end{proposition}
We check that when $S_1,S_2\in\partial\barVPM$, we have $d_H(S_1,S_2)=+\infty$ but $d_{H,\eps}(S_1,S_2)<+\infty$.
Thus enlarging $\VPM$ allows us to define a Hilbert distance $d_{H,\eps}$ which now admits boundedness for degenerate covariance and/or precision matrices.
This may prove useful in a number of scenario involving algorithms dealing with a mix of degenerate covariance and precision matrices.

Next, we may consider the full Gaussian family $\{N(\mu,\Sigma) \st \mu\in\bbR^n, \Sigma\in\PD(n) \}$ instead of zero-centered Gaussians by embedding  full Gaussian distributions into the SPD cone of higher dimension following the Calvo-Oller embedding~\cite{EllipticalDistance-CalvoOller-2002}:
$$
(\mu,\Sigma) \mapsto  \in \Sigma^+_\mu= \mattwotwo{\Sigma+\mu\mu^\top}{\mu}{\mu^\top}{1}\subset\PD(n+1).
$$
We then map $\Sigma^+_\mu$ into $\VPM(n+1)$ using $T_1$ transformation. 

Having straight lines being geodesics allow one to implement easily computational geometric primitives.
For example, we may extend the Badoiu and Clarkson iterative geodesic-cut algorithm for approximating the smallest enclosing ball~\cite{badoiu2003smaller,arnaudon2013approximating} (SEB).
Figure~\ref{fig:HilbertVPMSEB} reports a fine approximation of the Hilbert VPM SEB.  

We may also consider the VPM subdomain of correlation matrices, potentially degenerate.
For the SPD cone $\PD(n)$, the subdomain of correlation matrices is called an elliptope, an open bounded convex domain whose simplicial facial structure was studied in~\cite{tropp2018simplicial} and corresponding Hilbert geometry investigated in~\cite{nielsen2018clustering}.

We leave these extensions as well as the facial characterizations of $\VPM$ and  other related results for the next revision of this manuscript.

\begin{figure}%
\centering
\includegraphics[width=0.5\columnwidth]{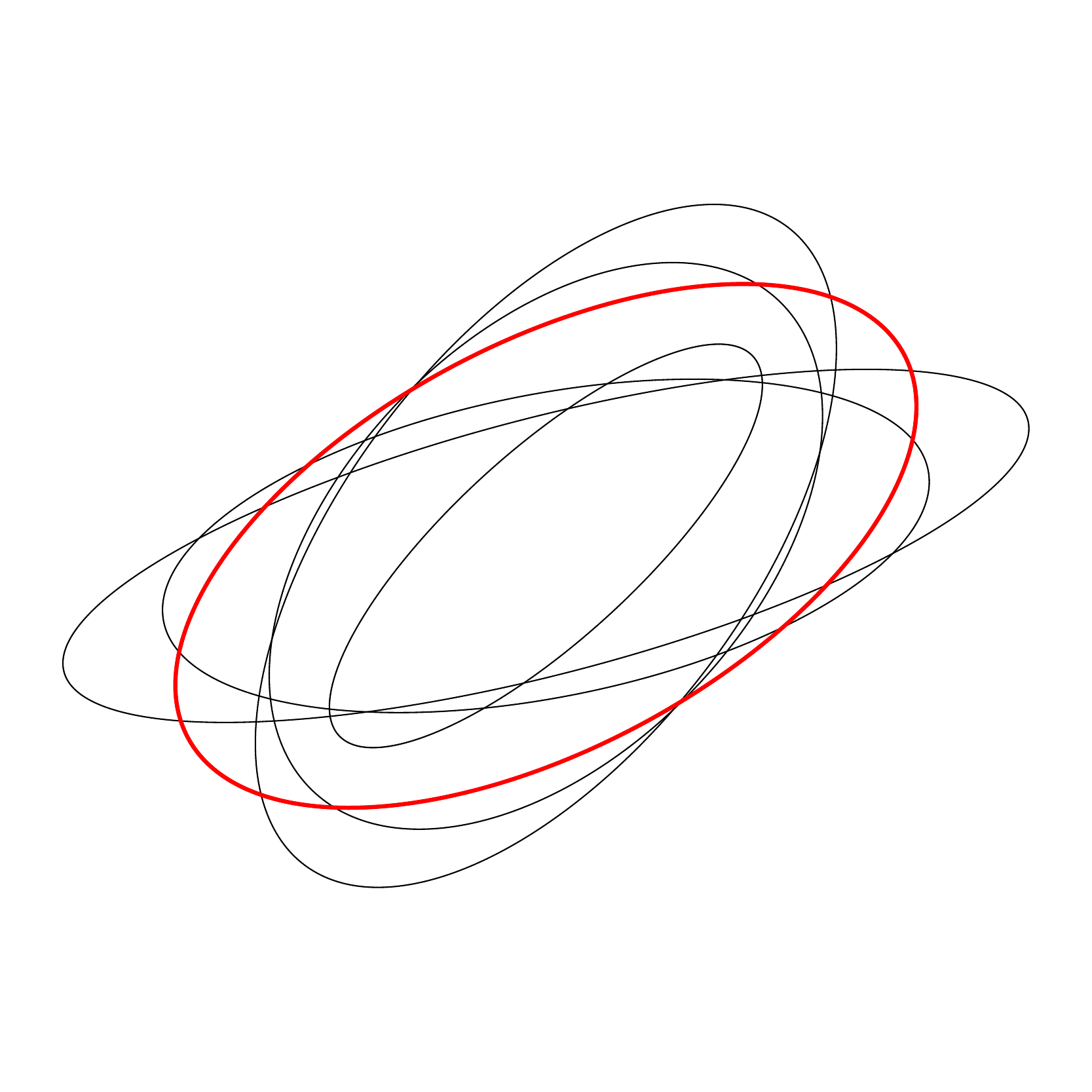}
\caption{Fine approximation of the smallest enclosing ball (SEB, in red) with respect to the Hilbert VPM distance in 2D of a set of 2D SPD matrices (or equivalently, 2D  zero-centered Gaussians, in black).}\label{fig:HilbertVPMSEB}
\end{figure}

\bibliographystyle{plain}
\bibliography{VPMBIB.bib}

\end{document}